\documentclass[11pt, onecolumn, draftcls]{IEEEtran}

\usepackage{enumerate,latexsym,amsmath,amssymb,amsfonts,verbatim}
\usepackage{dsfont}
\usepackage[final]{graphicx}
\usepackage{overpic}
\usepackage{url}
\usepackage{multirow}
\usepackage{array}
\usepackage{subfigure}
\usepackage{hyperref} 
\usepackage{mathrsfs}
\usepackage{cite}
\usepackage{setspace}

\newtheorem{theorem}{Theorem}[section]
\newtheorem{lemma}[theorem]{Lemma}

\newtheorem{corollary}[theorem]{Corollary}
\newtheorem{definition}{Definition}
\newtheorem{remark}{Remark}

\begin{document}
\doublespacing

\title{An Information Theoretic Study of  Timing Side Channels in Two-user Schedulers
(Draft)}
\author{
\IEEEauthorblockN{Xun~Gong,~\IEEEmembership{Student~Member,~IEEE,}
        Negar~Kiyavash,~\IEEEmembership{Member,~IEEE,}
        and~Parv~Venkitasubramaniam,~\IEEEmembership{Member,~IEEE}}\\
 \thanks{This work was supported in part by National Science Foundation through the grant CCF 10-65022, CCF 10-54937 CAR, and in part by Air Force through the grant FA9550-11-1-0016, FA9550-10-1-0573.
 This work was presented in part at ISIT'11. 
 
 X. Gong is with the Coordinated Science Laboratory and the Department of Electrical and Computer Engineering, University of Illinois at Urbana-Champaign, Urbana, IL 61801 USA (email: \url{xungong1@illinois.edu})
 
 N. Kiyavash is with  the Coordinated Science Laboratory  and the Department of Industrial and Enterprise Systems Engineering, University of Illinois at Urbana-Champaign, Urbana, IL 61801 USA (email: \url{kiyavash@illinois.edu})
 
 P. Venkitasubramaniam is with the Department of Electrical and Computer Engineering, Lehigh University, Bethlehem PA 18015 USA (email: \url{parv.v@lehigh.edu})
 }   
}

\maketitle

\begin{abstract}
Timing side channels in two-user schedulers are studied. 
When two users share a scheduler, one user may learn the other user's behavior from patterns of service timings.
We measure the information leakage of the resulting timing side channel in schedulers serving a legitimate user and a malicious attacker, 
using a privacy metric defined as the Shannon equivocation of the user's job density. 
We show that the commonly used first-come-first-serve (FCFS) scheduler provides no privacy as the attacker is able to to learn the user's job pattern completely. Furthermore, we introduce an scheduling policy, accumulate-and-serve scheduler, which services jobs from the user and attacker in batches after buffering them. 
The information leakage in this scheduler is mitigated at the price of service delays, and the maximum privacy is achievable when large delays are added. 
\end{abstract}

\IEEEpeerreviewmaketitle

\section{Introduction}

\PARstart{T}iming channels are created when information is transmitted in event timings. 
For instance, in packet networks, not only the packets' contents, but also their inter-arrival times can be used to carry information. 
Traditionally, timing channels are synonymous with covert channels, wherein parties are not allowed to communicate yet they do so~\cite{Lampson73}. Take the CPU scheduling channel in multi-level secure computer systems for example~\cite{millen89, Moskowitz96}. By modulating the number of quanta occupying the CPU, one process can signal messages covertly to another with a lower security level.  The recipient process is able to decode these messages from observing the CPU's busy periods. 
More recently, timing channels are also frequently exploited to implement  side channel attacks. 
Unlike covert channels, there is no active message sender in a side channel. Instead, a malicious process may passively or actively learn about the activities of a victim process by utilizing the timing evidence left on the shared resource. For instance, in the aforementioned CPU case, the number of CPU quanta required for completing a job implicitly reveals information about the process issuing it. A malicious process with access to the same CPU, can possibly infer unintended information about the underlying activities of the victim process. 
For example, if the victim process is running a decryption function,  the attacker may learn the encryption key from observing the time it takes for completing the decryption operation~\cite{kocher96, percival05}. 

Timing side channels are increasingly more perilous for user privacy as the result of more of our daily activities moving to networks where coupling of resources is inevitable. 
For example, in the software-as-a-service cloud computing platform deployed by Amazon, a server usually hosts jobs from several clients, which gives  malicious clients  the chance of probing workloads of neighbors~\cite{ristenpart09}. 
A timing side channel was recently discovered within home digital subscriber line (DSL) routers, using which an attacker learns a user's web traffic pattern~\cite{gong12, kadloor10}. 
This attack exploits the fact that packets downloaded in a DSL link are processed through an FCFS buffer, as shown in Figure~\ref{fig:threat_model}. The attacker Bob sends pings (Internet Control Message Protocol (ICMP) requests) to measure round trip times (RTT) for reaching the victim, Alice's, computer.
The ping requests along with Alice's web packets wait in the buffer to be processed. Thus, ping responses are delayed whenever Alice's applications download large volumes of traffic. 
Figure~\ref{fig:trafficvsrtt}  illustrates this scenario. Specifically,  Figure~\ref{fig:yahoo} depicts the RTTs of Bob's ping packets issued  every 10 ms, and Figure~\ref{fig:rtt} shows the data volume downloaded by Alice during the same time period.  It can be clearly seen that Bob's RTTs reveal the pattern of Alice's arrival process which may be further processed to identify the webpage Alice is browsing~\cite{gong12}.

\begin{figure}[t]
   \centering
   \includegraphics[width=0.7\columnwidth]{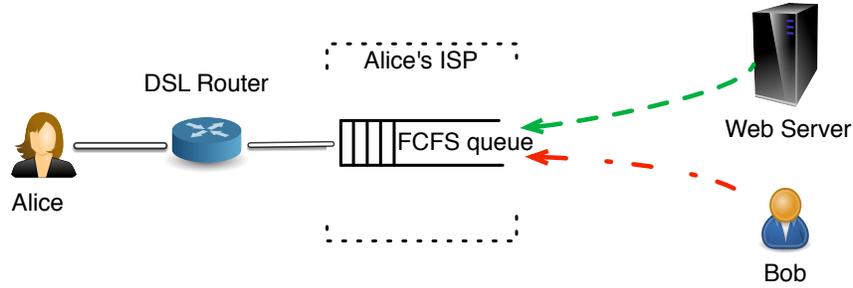} 
    \caption{A timing side channel in DSL routers. The user {\em Alice}'s download packets and the {\em Bob}'s ping packets share services from the DSL router. As a result, the timings of {\em Bob}'s packets convey information regarding sizes of Alice's packets.} 
   \label{fig:threat_model}
\end{figure}
\begin{figure}[t]
   \centering
    \subfigure[{\em Alice's} download pattern]{
    \label{fig:yahoo}
    \includegraphics[width=0.7\columnwidth]{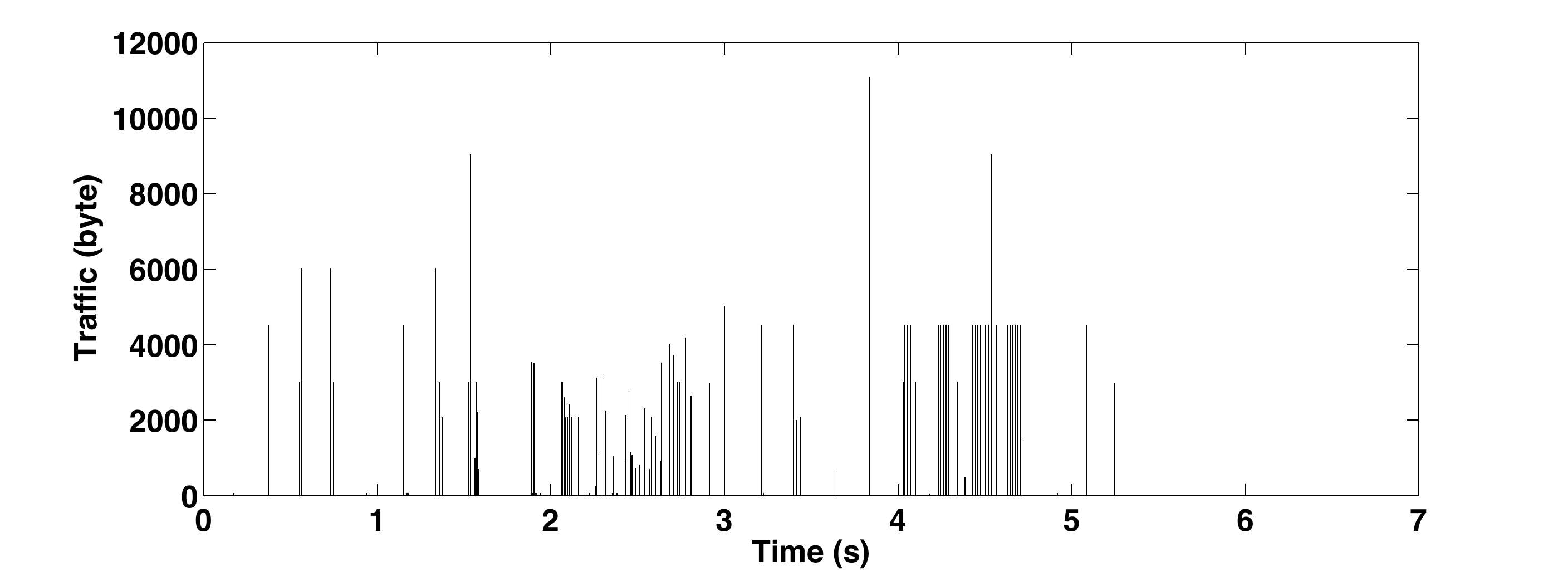}
      } \subfigure[RTTs measured by {\em Bob}]{
       \label{fig:rtt}
    \includegraphics[width=0.7\columnwidth]{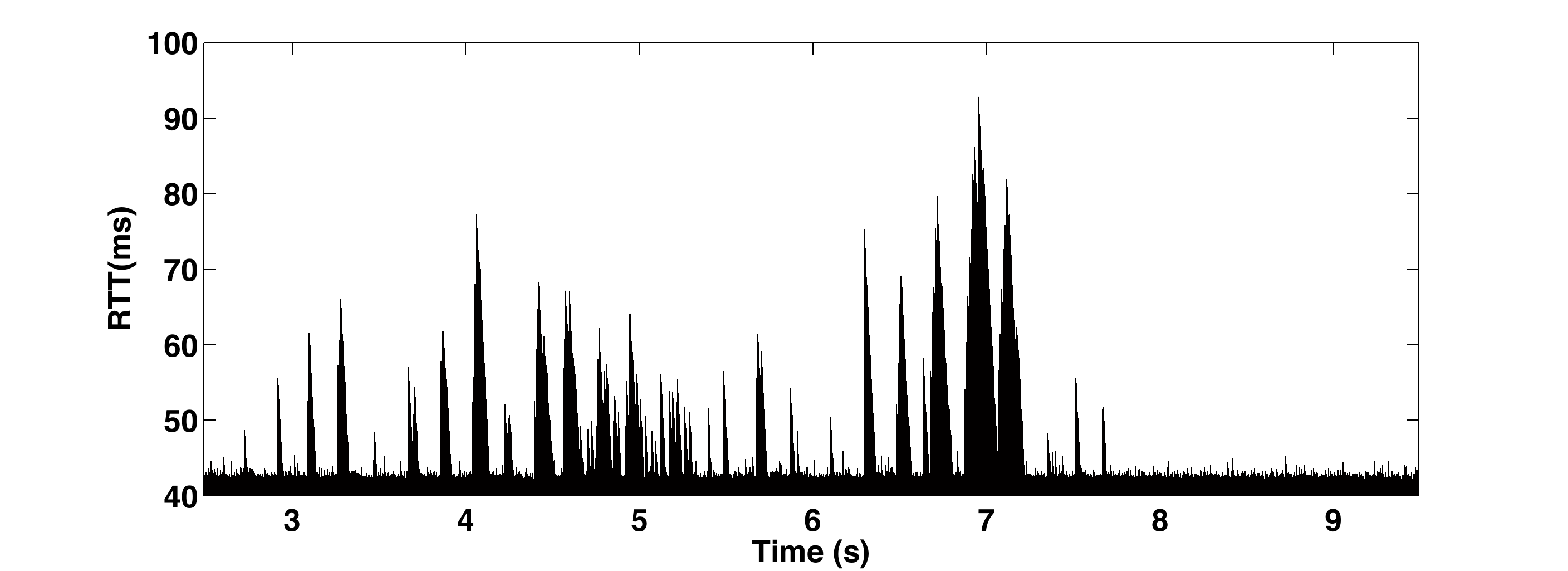}
      }
    \caption{
Information leakage through the timing side channel inside DSLs.    
(a) shows the total size of packets downloaded by Alice during every interval of 10\,ms.
(b) plots the RTT of Bob's ping that is issued every 10\,ms to Alice's computer.}
 \label{fig:trafficvsrtt}
\end{figure}

\subsection{Related Work}
Some noteworthy contributions in analyzing  the communication capacity of covert timing channels include~\cite{millen89, Moskowitz96, anantharam96, Wagner}.
Anantharam and Verd\a'u analyzed the communication capacity of a timing channel for a FCFS queue servicing jobs from one single arrival process~\cite{anantharam96}. 
The communication capacity  in their model depends on the service model; the minimum capacity was shown to be achieved by exponentially distributed service times. 
The communication capacity between job processes of a round-robin CPU scheduler was studied in~\cite{millen89, Moskowitz96}.
Millen proved the maximum timing channel information rate of a round-robin scheduler is  
$\log\left(\frac{1+\sqrt{5}}{2}\right)$ bits per quantum, achieved when the sender uniformly picks an arrival pattern for issuing jobs.
Additionally, techniques to mitigate covert channels were studied in ~\cite{kang93,kang96, Siva, Giles02}, where the main idea is to try to disrupt the communication among the processes by adding `dummy' service delays through an intermediate device, referred to as `pump' or `jammer'. 

We study the timing side channel between two users,  an attacker and a regular user, sending jobs that are scheduled through a shared server. Information  leakage in this side channel is determined by the scheduling policy. In this paper, we quantify information leakage using Shannon equivocation and analyze privacy of commonly used FCFS policy.
Similar studies under this model can be found in~\cite{kadloor12,kadloor13, gong11}, where  minimum-mean-square-error (MMSE) and attack-dependent metrics were used. 

Our main results are summarized in the following.
\begin{itemize}
\item We develop an information-theoretic framework for quantifying information leakage of timing side channels in schedulers using  Shannon's equivocation as a metric to access the privacy level provided by a scheduler (for details on Shannon's equivocation, refer~\cite{ash}).

\item We characterize the information leakage of a FCFS policy and show that the attacker learns the user's arrival pattern exactly if sufficient rate is available to him for sampling the queue. This demonstrates that FCFS is a poor policy in terms of preserving user's privacy despite its ease of implementation, high QoS, and ubiquity.

\item We suggest a policy, accumulate-and-serve, which trades off privacy and QoS (delay) by servicing jobs from the attacker and the user in separate batches buffered periodically. We prove that full privacy is achieved when large delays are added. 
\end{itemize}

The rest of the paper is organized as follows. 
A formal definition of our problem including the system model and metric are discussed in Section~\S\ref{sec:problem}.
Privacy of the FCFS scheduler is analyzed in Section~\S\ref{sec:fcfs}, 
The analysis for the privacy of the accumulate-and-serve policy follows in Section~\S\ref{sec:acc}.
Concluding remarks are presented in Section~\S\ref{sec:con}.

\section{Problem Formulation}
\label{sec:problem}

In this section, the problem formulation and notation are introduced.
Throughout the paper: bold script $\mathbf{A}$ denotes the infinite sequence  $\{A_1, A_2, \cdots\}$,  $\mathbf{A}^n$ denotes the sequence  $\{A_1, A_2, \cdots, A_n\}$, and $\mathbf{A}^{j}_{i}$ denotes the subsequence $\{A_{i}, A_{i+1}, \cdots, A_{j}\}$, where $j\geq i$. 

 \begin{figure}[t]
   \centering
   \includegraphics[width=0.5\columnwidth]{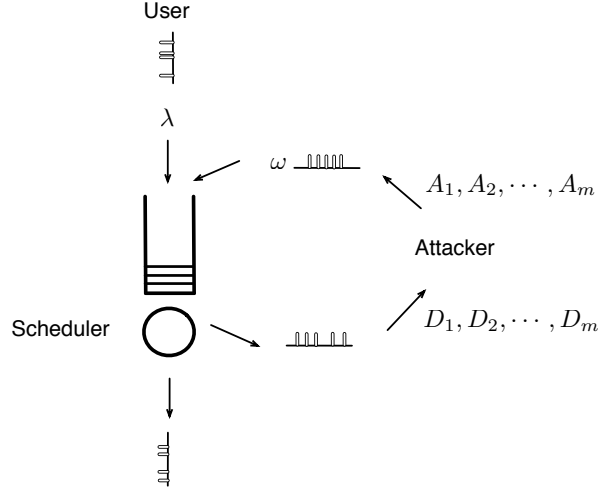} 
   \caption{
   A scheduler services  jobs from two arrival processes;  one is from a malicious  attacker who wants to probe the job pattern of the other, a legitimate user.  The attacker sends jobs to the scheduler to get knowledge of the queue status, based on which it infers the legitimate user's privacy. }
   \label{fig:sys_mod}
\end{figure}

\subsection{System Model}
\label{sec:sys_mod}

\begin{figure}[t] 
   \centering
   \includegraphics[width=0.7\columnwidth]{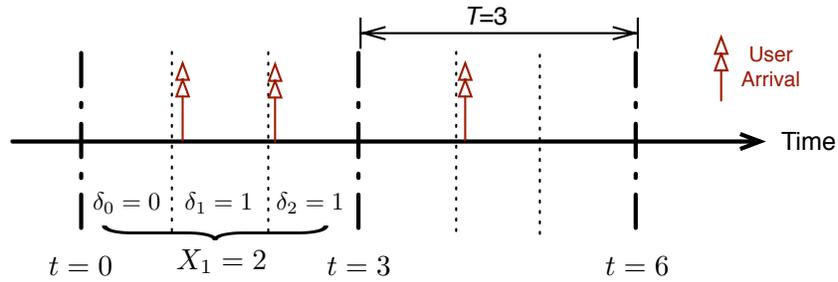} 
   \caption{The user's job pattern $\mathbf{X}=\{X_1, X_2, \cdots\}$ defined as the number of jobs in every clock period of $T$ time slots.
Here $T=3$,  $X_1=2$,  and $X_2=1$. } 
      \label{fig:metric}
\end{figure}

Figure~\ref{fig:sys_mod} depicts the shared scheduler which processes jobs from a regular user and an attacker in discrete time.  
At every time slot,  the user (and the attacker) can either issue one job or remain idle. All jobs have the same size and take one slot to get serviced.
We assume that the user's arrival process is Bernoulli with rate $\lambda$. 
Note that the difficulty in learning user's arrival pattern depends on the unpredictability of the pattern.
 It is easier for the attacker to learn the arrival pattern of the user if the user issues jobs in a predicable or regular pattern such as ON/OFF traffic with a fixed period.  On the other hand, the {\em Bernoulli process} is quite unpredictable as it is the maximum entropy discrete time stationary process for a fixed arrival rate (similar results for Poisson arrivals can be found in~\cite{mcfadden65}).
The attacker is allowed to send his jobs in any time slots as long as his long term rate  $\omega$
does not  exceed  $1-\lambda$, so as to avoid an unstable queue. Unlike in a denial of a service attack, in a side channel attack, the attacker does not benefit from overloading the server which results in dropping packets and hence loosing information.

The goal of the attacker is to learn the user's  job pattern $\mathbf{X}=\{X_1,X_2,\cdots\}$, as depicted in Figure~\ref{fig:metric}. 
\begin{definition}
The user's job pattern in the $k^{th}$ clock period is given by
\begin{equation}
X_k= \sum_{j=(k-1)T}^{kT-1}\delta_j,  \quad k=1,2,\cdots,
\label{eq:pattern_1}
\end{equation}
where $\delta_j\sim Bernoulli(\lambda), j=1,2,\cdots,$ labels the arrival event from the user at every time slot.
\label{def:pattern}
\end{definition}
The job pattern $\mathbf{X}$ presents a sampled view of the user's arrival process, with an observation every $T$ time slots. 
Among all the sampling sequences with rate $\frac{1}{T}$, the evenly-paced sampling captures the maximum information of the original Bernoulli process~\cite{xunDoc}. This implies that  $\mathbf{X}$ serves a proper objective for an attacker who wants to know as accurate information about user. 
The clock period $T$ sets the granularity of this side channel attack. 
A smaller value of $T$ indicates that the attacker intends to obtain a higher resolution view of the user's activity.
In the extreme case of $T=1$, the attacker wants to know learn whether a job was issued by the user in every single time slot.

\subsection{Privacy Metric}
\label{sec:metric}

We measure the user's privacy in the shared scheduler as the equivocation rate of his job pattern given the attacker's observations of his own jobs. 
Shannon equivocation is frequently used as a metric for information leakage in communication systems, such as the wiretap channel~\cite{Wyner75}.
Beside being a measure of uncertainty, equivocation provides a tight upper and lower bound for the minimum error probability~\cite{feder94}, which implies our proposed metric also bounds the error the attacker incurs in guessing user's job pattern. Such an error has been studied using a minimum-mean-square-error (MMSE) metric in~\cite{kadloor12, kadloor13}. 

Denote the arrival and departure times of attacker's jobs by $\mathbf{A}=\{A_1,A_2,\cdots\}$ and $\mathbf{D}=\{D_1,D_2,\cdots\}$ respectively. The attacker's arrival rate can be represented by $\omega =\underset{k\to\infty}{\lim}\frac{k}{A_k}$.
Suppose $m$ attacker's jobs were issued during the first $n$ clock periods, the uncertainty of the first $n$ job patterns to the attacker is then $H\left(\mathbf{X}^{n}| \mathbf{A}^m,\mathbf{D}^m\right)$. 
\begin{definition}
The user's {\em privacy} in a shared scheduler serving him and an attacker is given by
\begin{equation}
\begin{aligned}
\mathcal{P}^{T}&=\underset{\mathbf{A}:\omega<1-\lambda}{\min}\quad
\underset{n\to\infty}{\lim}\frac{H\left(\mathbf{X}^{n}| \mathbf{A}^m,\mathbf{D}^m\right)}{n},
\end{aligned}\label{eq:privacy}
\end{equation}
where $m=\underset{k}{\sup}\{k:A_k\leq nT\}$.
\label{def:privacy}
\end{definition}

$\mathcal{P}^T$ characterizes  the minimum equivocation of the user's job patterns for the best attack strategy satisfying rate restriction. 
The smaller the privacy $\mathcal{P}^T$, more the information that is leaked to the attacker through the timing side channel in the scheduler. 
A similar equivocation-based metric was proposed in~\cite{gong11}, where however the attacker's strategy is restricted as 
Bernoulli sampling and the metric was attack-dependent.

The value of $\mathcal{P}^{T}$ is largely determined by the policy  the scheduler uses. 
If the scheduler preassigns  fixed time slots to service each party, as in  TDMA, user's privacy is guaranteed because 
the service time of attacker's jobs is statistically independent with user's job patterns. 
In that case,
TDMA achieves the maximum privacy, as given by 
\begin{equation}
\mathcal{P}^{T}_{TDMA}= H({X}),
\label{eq:tdma}
\end{equation}
where $X$ is binomial $B(T,\lambda)$. 
%
However, such a policy results in idling of scheduler which  wastes resources and may add significant delays.  
Therefore, complete isolation of users' job processes is often not achievable in practice as the scheduler is required to maintain a certain level of QoS, such as average job delay. In such cases, a timing side channel is inevitable.
In the next section, we analyze the leakage of such a channel for FCFS scheduler which is widely deployed in practice.

\section{Information Leakage  in First-Come-First-Serve Scheduler}
\label{sec:fcfs}

FCFS is a simple scheduling policy commonly used in network systems. 
At each time slot, the FCFS scheduler services the job at the head of the queue. 
In the rest of this paper, for the sake of convenience of analysis, we assume that when both user and attacker issue a job in the same time slot, the attacker's job enters the queue first. 
As the scheduler never idles as long as the job queue is not empty, FCFS results in minimum average delay.\footnote{This is true when all jobs have the same size.}
However, FCFS exposes the queue length $q(\cdot)$ of the buffer to the attacker as the delay of the $i^{th}$ attacker's job is directly related to the number of jobs buffered before its arrival, i.e.,
\begin{equation}
q(A_i)=D_i-A_i-1,  
\label{eq:queue_service}
\end{equation}
where `1' accounts for the service time of the $i^{th}$ attacker's job itself. 
We subsequently show that by using a well designed attack strategy, the attacker can indeed significantly reduce the timing privacy of the user.

\subsection{Attack Strategy}
\label{sec:strategy}

Recall the attacker's objective is to learn the user's job pattern, i.e., the number of user's arrivals within each clock period $T$.
Therefore, given~\eqref{eq:queue_service}, the attacker should issue jobs at times $t=0, T, 2T,\cdots,$ to know the queue lengths on the clock period boundaries, which are essential to accurately estimating user's job pattern. 
Based on this observation, we design an attack strategy (Figure~\ref{fig:attack_strategy}), where
the attacker's jobs are of two types:
\begin{itemize}
\item {\em Type-I} jobs are issued on boundaries of clock periods, $t=0, T, 2T,\cdots$; 

\item {\em Type-II} jobs are issued on slots inside a clock period according to a Bernoulli process.
We use the remaining rate after issuing {\em Type-I} jobs to issue them. The probability of having one {\em Type-II} job in each slot within a clock period is 
$\frac{\omega T-1}{T-1}.
\label{eq:attack_rate}$
\end{itemize}
The purpose of issuing {\em Type-II} jobs is to ensure the queue is not empty. This reason for this becomes apparent in~\S\ref{sec:fcfs_res}.

\begin{figure}[t]
   \centering
   \includegraphics[width=0.7\columnwidth]{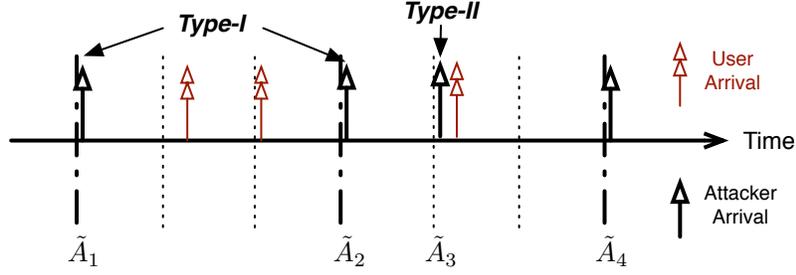} 
   \caption{Our attack strategy: the attacker issues  a {\em Type-I}  job on every clock tick and a {\em Type-II} job on each slot between clock ticks with probability of $\frac{\omega T-1}{T-1}$. }
   \label{fig:attack_strategy}
\end{figure}

The attack strategy in Figure~\ref{fig:attack_strategy} is not  feasible when $T$ is too small, as the attacker's rate is bounded by $1-\lambda$.
This poses an intrinsic limit on how much the attacker can learn from the side channel. More precisely, the attacker can not expect to learn user's job pattern at a resolution finer than his maximum sampling rate. 
For this reason, for the rest of our analysis, we consider learning the job pattern within the feasible resolution of the attacker; i.e., $T\geq \left \lfloor \frac{1}{1-\lambda} \right \rfloor$.

\subsection{An Upper Bound on Privacy}
\label{sec:upperbound}
The  attack strategy described in~\S\ref{sec:strategy} guarantees certain level of information gain for the attacker, and therefore sets an upper bound on the user's privacy in this side channel. 
Denote in this attack strategy $\mathbf{\tilde{A}}=\{\tilde{A}_1,\tilde{A}_2,\cdots\}$ the arrival times of attacker's jobs, and 
$\mathbf{\tilde{D}}=\{\tilde{D}_1,\tilde{D}_2,\cdots \}$ the corresponding departure times.
 From~\eqref{eq:privacy}, we know
\begin{equation}
\mathcal{P}^{T}\leq \underset{\omega:\omega<1-\lambda}{\min}\quad
\underset{n\to\infty}{\lim}  \frac{ H\left(\mathbf{X}^{n}| \mathbf{\tilde{A}}^{\tilde{m}},\mathbf{\tilde{D}}^{\tilde{m}}\right)}{n},
\label{eq:privacy_1}
\end{equation}
where $\tilde{m}=\sup\{k:\tilde{A}_k\leq nT\}$ is the total number of jobs sent by the attacker over period $nT$.

We first analyze the queuing stability of the scheduler under this attack strategy.
\begin{lemma}
When $\lambda+\omega<1$, the queue lengths observed at clock period boundaries $\{q(iT)\}, i=0, 1,\cdots,$ form a positive recurrent Markov chain.
\label{lem:stability}
\end{lemma}
\begin{proof}
See~Appendix~\ref{app:system_dynamic}.
\end{proof}

\begin{corollary}
When $\lambda+\omega<1$, the pairs 
$\left\{\tilde{A}_k, q(\tilde{A}_k); (i-1)T\leq \tilde{A}_k \leq iT\right\}, i=1,2,\cdots,$ form a positive recurrent Markov chain.  
\label{cor:stability}
\end{corollary}
\begin{proof}
See~Appendix~\ref{app:system_dynamic_2}.
\end{proof}

\begin{remark}
Lemma~\ref{lem:stability} and Corollary~\ref{cor:stability} demonstrate the existence of steady state of the scheduler's queue length. 
This implies the convergence of the limit of the equivocation rate in~\eqref{eq:privacy_1}, using which we derive an upper bound on user's privacy. 
\end{remark}

\begin{lemma}
Consider the FCFS scheduler with the total  job rate $\lambda+\omega<1$, where $\lambda$ and $\omega$ denote the user's arrival rate and attacker's arrival rate respectively. 
The user's privacy is upper-bounded by\begin{equation}
\begin{aligned}
\mathcal{P}^{T}\leq \underset{\omega:\omega<1-\lambda}{\min} \mathbb{E}_{s}\left[\sum_{i=1}^{s+1} H\left(\mathcal{X}_i|\tau_i, Q_{i},Q_{i+1}\right)\right],
\label{eq:privacy_convergence}
\end{aligned}
\end{equation}
where $s$ is binomial $B\left(T-1,\frac{\omega T-1}{T-1}\right)$, $\tau_i$ is geometric $G\left(\frac{\omega T-1}{T-1}\right)$, and $\mathcal{X}_i$ is binomial $B(\tau_i,\lambda)$ for  $i=1,2,\cdots,s+1$. Moreover, \begin{equation}
\sum_{i=1}^{s+1} \tau_i=T, 
\end{equation} and
$Q_1$, $Q_{s+2}$ are identically distributed and
 \begin{equation}Q_{i+1} = \left(Q_{i} +\chi_i+1-\tau_i \right)_{+}, \quad \forall i=1,2,\cdots,s+1.
 \label{eq:Q}
\end{equation}
\label{cor:upper_bound}
\end{lemma}
\begin{proof}
Expand the conditional entropy in~\eqref{eq:privacy_1} using the entropy chain rule:
\begin{equation}
\begin{aligned}
H\left(\mathbf{X}^n\big| \mathbf{\tilde{A}}^{\tilde{m}},\mathbf{\tilde{D}}^{\tilde{m}}\right) &=  \sum_{k=1}^{{n}} H\left(X_k\big|\mathbf{X}^{k-1},  \mathbf{\tilde{A}}^{\tilde{m}},\mathbf{\tilde{D}}^{\tilde{m}} \right)\\
&\overset{(a)}{=}  \sum_{k=1}^{{n}} H\left(X_k\big|\mathbf{X}^{k-1},  \mathbf{\tilde{A}}^{\tilde{m}}, q(\tilde{A}_1), \cdots, q(\tilde{A}_{\tilde{m}}) \right)
\end{aligned}
\label{eq:cond_entropy_chain}
\end{equation}
where $(a)$ follows from~\eqref{eq:queue_service}.

Denote the total number of user's jobs sent during $[\tilde{A}_i,\tilde{A}_{i+1})$ (between two consecutive attack jobs) by $\hat{X}_{i}$. Note that 
\begin{equation}
\hat{X}_i= \sum_{j=\tilde{A}_i}^{\tilde{A}_{i+1}-1}\delta_j \quad, i=1,2,\cdots,
\label{eq:pattern_2}
\end{equation} 
wherein $\delta_j$ is a $Bernoulli(\lambda)$ indicator of whether user issued a job at the $j^{th}$ time slot. 

From~\eqref{eq:pattern_1}, we know
\begin{equation}
X_k=\underset{i:(k-1)T\leq \tilde{A}_i < kT}{\sum} \hat{X}_i, \quad k= 1,2,\cdots.
\label{eq:pattern_connect}
\end{equation}
Plug~\eqref{eq:pattern_connect} into~\eqref{eq:cond_entropy_chain}, 
\begin{equation}
\begin{aligned}
H\left(\mathbf{X}^n\big| \mathbf{\tilde{A}}^{\tilde{m}},\mathbf{\tilde{D}}^{\tilde{m}}\right)&\overset{}{=}  \sum_{k=1}^{{n}} H\left(\underset{i:(k-1)T\leq \tilde{A}_i < kT}{\sum} \hat{X}_i\Big| \mathbf{X}^{k-1},  \mathbf{\tilde{A}}^{\tilde{m}}, q(\tilde{A}_1), \cdots, q(\tilde{A}_{\tilde{m}}) \right)\\
&\overset{(b)}{=}\sum_{k=1}^n H\left(\underset{i:(k-1)T\leq \tilde{A}_i < kT}{\sum} \hat{X}_i\Big| \tilde{A}_i, q(\tilde{A}_i); (k-1)T\leq \tilde{A}_i \leq kT\right )
\end{aligned}
\label{eq:cond_entropy_chain_2}
\end{equation}
where $(b)$ holds because  the queue length update equation at the clock boundaries is given by 
\begin{equation}
q(\tilde{A}_{i+1})=\left(q(\tilde{A}_{i})+1+\hat{X}_{i}-(\tilde{A}_{i+1}-\tilde{A}_{i})\right)_{+}, i=1,2,\cdots.
\label{eq:q_update_2}
\end{equation}

From Corollary~\ref{cor:stability}, we know  that $\left\{\tilde{A}_i, q(\tilde{A}_i), (k-1)T\leq \tilde{A}_i \leq kT\right\}$, $k=1,2,\cdots,$ form a positive recurrent Markov chain.
Hence, the conditional entropy term in the sum of the last line in~\eqref{eq:cond_entropy_chain_2} converges as $k\to\infty$, with the limit determined by the stationary distribution of state $\{\tilde{A}_i, q(\tilde{A}_i); (k-1)T\leq \tilde{A}_i \leq kT\}$. 
Assume this chain is in the stationary state, and let
\begin{itemize}
\item[$s$] $=\Big|\left\{i: (k-1)T< \tilde{A}_i < kT\right\}\Big|$ be the the number of {\em Type-II} attack jobs issued in a clock period $T$. Then $s \sim B\left(T-1,\frac{\omega T-1}{T-1}\right)$ (for the attack strategy defined in \S\ref{sec:strategy});
\item[$\tau_i$]  $\sim G\left(\frac{\omega T-1}{T-1}\right), i=1,2,\cdots,s+1$ denote inter-arrival time of attacker's jobs in a clock period $T$.
Clearly, sum of these inter-arrival times equals $T$; 
\item[$\mathcal{X}_i$] $\sim B(\tau_i,\lambda), i=1,2,\cdots,s+1,$ be the number of user's jobs arriving between every pair of consecutive attacker's jobs; 
\item[$\mathbf{Q}^{s+2}$]  denote the queue lengths seen by the total $s+2$ attacker's jobs sent in $[(k-1)T, kT]$ ($s$ of {\em Type-II} and 2 of {\em Type-I}).
The queue length in sequence $\mathbf{Q}^{s+2}$  updates following from~\eqref{eq:Q}. Moreover, in the stationary state, queue lengths at clock period boundaries--$Q_1$, $Q_{s+2}$--have identical distribution. 
\end{itemize}

Then, we can write the  limit of the conditional entropy  in~\eqref{eq:cond_entropy_chain_2}  as 
\begin{equation} 
\begin{aligned}
\lim_{k \to \infty}&H\left(\underset{i:(k-1)T\leq \tilde{A}_i \leq kT}{\sum} \hat{X}_i\Big| \tilde{A}_i, q(\tilde{A}_i); (k-1)T\leq \tilde{A}_i \leq kT\right )
= \mathbb{E}_{s}\left[H\left(\sum_{i=1}^{s+1}\mathcal{X}_i  \Big| \mathbf{\tau}^{s+1}, \mathbf{Q}^{s+2} \right)\right].
\label{eq:entropy_limit_1}
\end{aligned}
\end{equation}

Substituting~\eqref{eq:cond_entropy_chain_2} back into~\eqref{eq:privacy_1}, and applying~\eqref{eq:entropy_limit_1} and Ces\`{a}ro mean theorem \cite[Theorem~4.2.3]{thomas91}, 
\begin{equation}
\begin{aligned}
\mathcal{P}^{T} & \leq \underset{\omega:\omega<1-\lambda}{\min} \mathbb{E}_{s}\left[H\left(\sum_{i=1}^{s+1}\mathcal{X}_i  | \mathbf{\tau}^{s+1}, \mathbf{Q}^{s+2} \right)\right]\\
&\leq \underset{\omega:\omega<1-\lambda}{\min}\mathbb{E}_{s}\left[\sum_{i=1}^{s+1}H\left(\mathcal{X}_i|\mathbf{\tau}^{s+1}, \mathbf{Q}^{s+2}\right)\right]\\
&\overset{(c)}{=}\underset{\omega:\omega<1-\lambda}{\min}\mathbb{E}_{s}\left[\sum_{i=1}^{s+1}H\left(\mathcal{X}_i|\tau_i, Q_{i},Q_{i+1} \right)\right]. 
\end{aligned}
\label{eq:upper_bound}
\end{equation}
where $(c)$ is based on the queue length update equation~\eqref{eq:Q}.
\end{proof}

\subsection{FCFS Provides Zero Limiting Privacy}
\label{sec:fcfs_res}
We next show that the bound in Lemma~\ref{cor:upper_bound} converges to 0 as attacker's job rate approcaches $1-\lambda$. Therefore, FCFS scheduler provides no privacy. 

\begin{lemma}
If the current arrival of the attacker  sees a non-empty queue, he learns the exact number of jobs the user has issued between the attacker's current and previous jobs, or
\begin{equation}
H\left(\mathcal{X}_i|\tau_i, Q_{i},Q_{i+1} \right)= 0, \quad \text{ if } Q_{i+1}>0, i=1,2,\cdots, s+2,
\label{eq:consecutive}
\end{equation}
for all $0 \leq s\leq T-1$.
\label{lem:consecutive} 
\end{lemma}
\begin{proof}
From~\eqref{eq:Q} when $Q_{i+1}>0$, we have \begin{equation}\mathcal{X}_i=Q_{i+1}-Q_{i}+\tau_i-1, \end{equation} 
 which implies $H\left(\mathcal{X}_i|\tau_i, Q_{i},Q_{i+1} \right)=0$. 
\end{proof}

\begin{remark}
The intuition provided by Lemma~\ref{lem:consecutive} is that when the queue is nonempty, the attacker does not miss the legitimate user's arrivals. Therefore, the attacker has the incentive to issue as many jobs as possible to create a queue. This is the motivation for issuing  {\em Type-II} jobs in our attack strategy of \S\ref{sec:strategy}. 
\end{remark}
When the attacker makes full use of available rate, he always sees a busy scheduler, as stated in the next lemma.

\begin{lemma}
In a FCFS scheduler, an attacker issuing jobs according to time sequence $\mathbf{\tilde{A}}$ as described in~\S\ref{sec:strategy} at the maximum available rate rarely sees an empty queue, or
\begin{equation}
\lim_{\omega\to 1-\lambda}Pr(Q_i > 0)=1, \quad  i=1,2,\cdots, s+2,
\label{eq:queue_length}
\end{equation}
for all $0 \leq s\leq T-1$.
\label{lem:queue_length}
\end{lemma}
\begin{proof}
From~\eqref{eq:Q}, if the attacker sees a queue length greater than $T-1$ at the clock boundary, then all its jobs arriving in the  the following clock period will experience a nonempty queue. 
Hence,  a sufficient statement for~\eqref{eq:queue_length}  to hold is that  
\begin{equation} \lim_{\omega\to 1-\lambda}Pr(Q_1 \geq T-1)=1.\label{eq:Q_limit}\end{equation}
See Lemma~\ref{lem:app_1} in Appendix~\ref{app:queue_length} for a proof of~\eqref{eq:Q_limit}.
\end{proof}

Based on Lemma~\ref{cor:upper_bound}, Lemma~\ref{lem:consecutive}, and Lemma~\ref{lem:queue_length}, we now present the main theorem characterizing the privacy behavior of FCFS schedulers. 

\begin{theorem}
The FCFS scheduler provides no privacy of user's job patterns, or
\begin{equation}
\mathcal{P}^{T}=0.
\end{equation}
\label{theo:main}
\end{theorem}
\begin{proof}
From~\eqref{eq:privacy_convergence},
\begin{equation}
\begin{aligned}
\mathcal{P}^{T}&\leq  \underset{\omega:\omega<1-\lambda}{\min}\mathbb{E}_{s}\left[\sum_{i=1}^{s+1}H\left(\mathcal{X}_i|\tau_i, Q_{i},Q_{i+1} \right)\right]\\
&\overset{}{\leq} \lim_{\omega\to 1-\lambda} \mathbb{E}_{s}\left[\sum_{i=1}^{s+1} H\left(\mathcal{X}_i|\tau_i, Q_{i},Q_{i+1}\right)\right]\\
&\overset{(d)}{=} \lim_{\omega\to 1-\lambda} \mathbb{E}_{s}\left[\sum_{i=1}^{s+1} Pr(Q_{i+1}=0)  H\left(\mathcal{X}_i|\tau_i, Q_{i},Q_{i+1}=0\right)\right]\\
&\overset{}{=} \lim_{\omega\to 1-\lambda} \mathbb{E}_{s}\left[\sum_{i=1}^{s+1} \left(1-Pr(Q_{i+1}>0)\right)  H\left(\mathcal{X}_i|\tau_i, Q_{i},Q_{i+1}=0\right)\right]\\
&\overset{(e)}{=}0
\label{eq:temp3}
\end{aligned}
\end{equation}
where $(d)$ follows from~\eqref{eq:consecutive}, and $(e)$ results from~\eqref{eq:queue_length}.
Since $\mathcal{P}^{T}$ defined in~\eqref{eq:privacy}, cannot be negative, we must have $\mathcal{P}^{T}=0$. 
\end{proof}

\section{An Accumulate-and-Serve Scheduler}
\label{sec:acc}

As we showed in the previous section, FCFS preserves little privacy despite its QoS and complexity advantages. In this section, we propose a new policy, accumulate-and-serve that mitigates the side channel information leakage by adding service delays. 
This scheduling policy is similar to the periodic dump jammer previously proposed to mitigate covert channels~\cite{Giles02}. 
By buffering jobs periodically and servicing batches belonging to different users separately, the correlation between the attacker's departure process and user's arrival process is greatly reduced. As a result, the accumulate-and-serve scheduler gives the attacker a coarser view of the user's job patterns,  compared to a FCFS scheduler.

Our accumulate-and-serve scheduler  works as follows: 
time slots are divided into intervals with length of $T_{acc}$. 
The scheduler accumulates all jobs that have arrived during an interval into two batches; one consisting of the user's jobs  and the other containing all jobs from the attacker.
Then the scheduler starts servicing these two batches starting at the next available time slot (after completing all previously scheduled jobs).
The order at which the user and the attacker get served is fixed for all the accumulate intervals.

An example of accumulate-and-serve scheduler is shown in Figure~\ref{fig:acc}, where the accumulate interval is set as $T_{acc}=9$.
In Figure~\ref{fig:acc},  the scheduler first waits for $9$ time slots, and then starts processing the accumulated jobs, 4 from the user and 3 from the attacker, at $t=9$ in two batches. The service order in this example is giving priority to the user's job batches. 

\begin{figure}[t]
   \centering
   \includegraphics[width=\columnwidth]{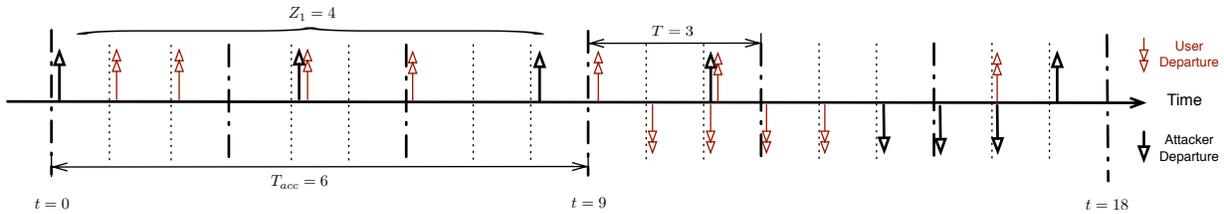} 
   \caption{An example of the accumulate-and-serve policy. At the beginning of each interval, the scheduler services immediately all jobs which arrived during the previous interval in two batches; first one user and then the other. In this example, the accumulate interval $T_{acc}=3T$, and the user's jobs are serviced first.}
   \label{fig:acc}
\end{figure}

\subsection{The Limitation of Attacks on Accumulate-and-Server Schedulers}

Under the accumulate-and-serve policy, the correlation between the user and attacker's processes is only through the size of 
 job batches. 
Therefore, the attacker can at most  learn the total size of user's jobs in each batch, and not the arrival pattern inside the accumulate period. In other words, the accumulate interval $T_{acc}$ sets an upper bound on the resolution to which the attacker can learn user's job pattern.

Denote the user's job pattern within one accumulate interval by $Z_k$,
\begin{equation}
Z_k= \sum_{j=(k-1)T_{acc}}^{kT_{acc}-1}\delta_j, \quad k=1,2,\cdots.
\label{eq:def_z}
\end{equation}
\begin{lemma}
In an accumulate-and-serve scheduler  serving a user and an attacker, the attacker's observation, the sequence $\mathbf{Z}$ and user's job pattern  form a Markov chain, i.e., 
\begin{equation}
\mathbf{A},\mathbf{D}\to \mathbf{Z} \to \mathbf{X}
\label{eq:acc_chain}
\end{equation}
\label{lem:acc_chain}
\end{lemma}
\begin{proof}
Assume that the first attacker's job in the $k^{th}$ accumulate interval arrives at time $A_{k^*}$, its departure time depends on whether the scheduler  gives priority to the user's job batch or the attacker's.
If it is the former, we have
\begin{equation}
D_{k^*} = \max\{D_{k^{*}-1}, kT_{acc}\}+Z_k+1,
\label{eq:head_batch_1}
\end{equation}
where $D_{k^{*}-1}$ is the time the services of all previously accumulated jobs are completed, and $Z_k$ is the total service time of user's job batch accumulated in the current $k^{th}$ interval. 
If attacker's job batch receives services first, then the scheduler finishes all jobs from previous accumulate intervals by $D_{k^{*}-1}+Z_{k-1}$, and starts immediately to serve the newly buffered jobs from the attacker, in which case
\begin{equation}
D_{k^*} = \max\{D_{k^{*}-1}+Z_{k-1}, kT_{acc}\}+1.
\label{eq:head_batch_2}
\end{equation}

Once $D_{k^*}$ is determined, the rest of the jobs in the $k^{th}$ batch of attacker are serviced back to back. 
Equations \eqref{eq:head_batch_1} and~\eqref{eq:head_batch_2} imply that the sequence $\mathbf{Z}$ is a sufficient statistic of $\mathbf{X}$ to generate the departure times of the attacker. Thus \eqref{eq:acc_chain} holds. 
\end{proof}

\begin{remark}
Lemma~\ref{lem:acc_chain} imposes an upper bound on the information leakage or equivalently a lower bound on the privacy. Specifically, it implies that  the attacker learns no more than information about $\mathbf{X}$ than what is contained in the sequence $\mathbf{Z}$. Note that from~\eqref{eq:def_z}, $H(\mathbf{Z})$ is a monotonically decreasing function of the accumulate interval $T_{acc}$. 
Therefore the scheduler can mitigate the leakage by picking a large accumulate period $T_{acc}$ albeit at the price of delay. 
\end{remark}

\subsection{A Lower Bound on Privacy}

The lower bound on privacy provided by the accumulate-and-serve scheduler is given in the following theorem.

\begin{theorem}
In an accumulate-and-serve scheduler with $T_{acc}>T$, the user's privacy is lower bounded by
\begin{equation}
\mathcal{P}^{T}\geq \left(1-\frac{T}{T_{acc}}+\frac{T}{T_{L}}\right)H(X) - \frac{T\cdot H\left(\sum_{i=1}^{\left\lfloor\frac{T_{acc}}{T}\right\rfloor}X_i\right)}{T_{acc}},
\label{eq:acc_bound}
\end{equation}
where $T_{L}=lcm\left(T,T_{acc}\right)$, and $X, X_1,\cdots, X_{\left\lfloor\frac{T_{acc}}{T}\right\rfloor}$ are all i.i.d. binomial $B(T,\lambda).$
\label{theo:acc_bound}
\end{theorem}
\begin{proof}
We first prove the case of  $T_{acc}=lT$, wherein $l$ is a positive integer, in which case the bound in~\eqref{eq:acc_bound} reduces to
\begin{equation}
\mathcal{P}^{T}\geq H(X) - \frac{H\left(\sum_{i=1}^{l}X_i\right)}{l}.
\label{eq:acc_bound_2}
\end{equation}

Applying the Markov chain of~\eqref{eq:acc_chain} to the privacy definition of~\eqref{eq:privacy}, and considering the equivocation of the first $nT_{acc}$ time slots, we have
\begin{equation}
\mathcal{P}^{T}\geq \lim_{n\to \infty} \frac{H\left(\mathbf{X}^{nl}|\mathbf{Z}^{n}\right)}{nl}.
\label{eq:acc_bound_2}
\end{equation}

Next apply the entropy chain rule to the conditional entropy of~\eqref{eq:acc_bound_2}:
\begin{equation}
\begin{aligned}
H\left(\mathbf{X}^{nl}\big|\mathbf{Z}^{n}\right)&=\sum_{i=1}^{n} H\left(\mathbf{X}_{(i-1)l+1}^{il}\big|\mathbf{X}^{(i-1)l},\mathbf{Z}^{n}\right)\\
&\overset{(f)}{=} \sum_{i=1}^{n} H\left(\mathbf{X}_{(i-1)l+1}^{il}\big|{Z}_i\right)\\
&\overset{(g)}{=} \sum_{i=1}^{n}H\left(\mathbf{X}_{(i-1)l+1}^{il}\Big|\sum_{j=1}^{l}X_{(i-1)l+j} \right)\\
&\overset{}{=}\sum_{i=1}^{n} \left( H\left(\mathbf{X}_{(i-1)l+1}^{il},\sum_{j=1}^{l}X_{(i-1)l+j}\right) -H\left(  \sum_{j=1}^{l}X_{(i-1)l+j} \right) \right)\\
&\overset{}{=}\sum_{i=1}^{n} \left( H\left(\mathbf{X}_{(i-1)l+1}^{il}\right) -H\left(  \sum_{j=1}^{l}X_{(i-1)l+j} \right) \right),
\end{aligned}
\label{eq:acc_chain_rule}
\end{equation}
where 
$(f)$, $(g)$ both follow from the fact that user's jobs within an accumulate interval consist of job patterns among $l$ $T$-clock periods, as given by
\begin{equation}
Z_{i}=\sum_{j=1}^{l} X_{(i-1)l+j}, \quad i=1,2,\cdots \frac{n}{l},
\end{equation}
which results from \eqref{eq:def_z} and~\eqref{eq:pattern_1}. 

Substituting~\eqref{eq:acc_chain_rule} back into~\eqref{eq:acc_bound_2}, and we have
\begin{equation}
\begin{aligned}
\mathcal{P}^{T}&\geq \frac{\sum_{i=1}^{n} \left( H\left(\mathbf{X}_{(i-1)l+1}^{il}\right) -H\left(  \sum_{j=1}^{l}X_{(i-1)l+j} \right) \right)}{nl}\\
& \overset{(h)}{=} H(X) - \frac{H\left(\sum_{i=1}^{l}X_i\right)}{l}
\end{aligned}
\end{equation}
where $X$ is  binomial $B(T,\lambda)$, and $(h)$ follows from the fact that $X_i$'s are i.i.d. binomial $B(T,\lambda)$.

See Appendix~\ref{app:complement} for the proof when $T$ does not divide $T_{acc}$. 
\end{proof}

\begin{theorem}
As the accumulate interval $T_{acc}\to \infty$, the user's privacy converges to
\begin{equation}
\mathcal{P}^{T}=H(X),
\label{eq:corollary}
\end{equation}
where $X$ is binomial $B(T,\lambda)$.
\end{theorem}
\begin{proof}
Consider the limit of each term in the bound of~\eqref{eq:acc_bound}, we have
\begin{equation}
\lim_{T_{acc}\to \infty}  \left(1-\frac{T}{T_{acc}}+\frac{T}{T_{L}}\right)H(X)  = H(X)
\label{eq:limit_1}
\end{equation}
and 
\begin{equation}
\begin{aligned}
 \lim_{T_{acc}\to\infty}\frac{T\cdot H\left(\sum_{i=1}^{\left\lfloor\frac{T_{acc}}{T}\right\rfloor}X_i\right)}{T_{acc}}\overset{(i)}{\leq} \lim_{T_{acc}\to\infty}\frac{\log_2\left(2\pi el'T\lambda(1-\lambda)\right)+O\left(\frac{1}{l'T}\right)}{l'}=0
\end{aligned}
\label{eq:limit_2}
\end{equation}
where $l'=\left\lfloor\frac{T_{acc}}{T}\right\rfloor$, and $(i)$ follows from the approximation for binomial entropy in~\cite[Theorem 3]{frank94}. 

Substituting~\eqref{eq:limit_1} and~\eqref{eq:limit_2} back to~\eqref{eq:acc_bound_2}, and taking the limit $T_{acc}\to \infty$, 
we have $\mathcal{P}^{T} \geq H(X).$
However from~\eqref{eq:tdma} we know that the privacy is upper-bounded by
$\mathcal{P}^{T} \leq  H(X).
\label{eq:acc_upper_bound}$
This completes the proof.
\end{proof}

\begin{figure}[t]
   \centering
   \includegraphics[width=0.7\columnwidth]{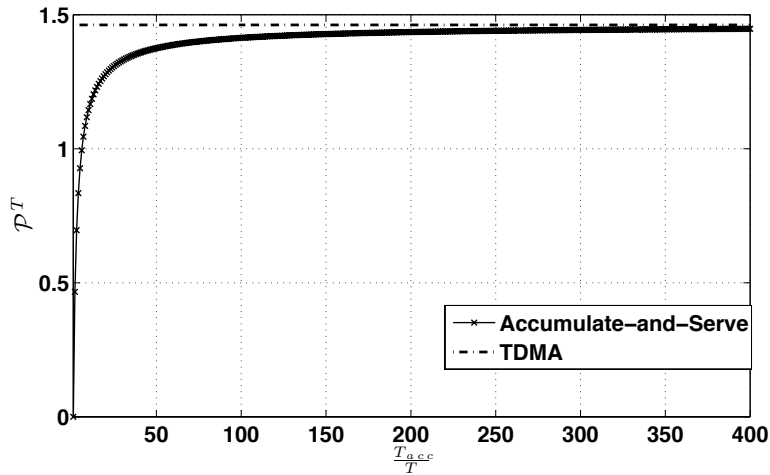} 
   \caption{A lower bound on privacy of the accumulate-and-serve scheduler, for $\lambda=0.4$, and $T=2$.}
   \label{fig:acc_bound}
\end{figure}

Figure~\ref{fig:acc_bound} illustrates the bound  in~\eqref{eq:acc_bound_2}. 
Not surprisingly,  the guaranteed privacy of the scheduler increases with the accumulate interval $T_{acc}$. When $T_{acc}$ is large enough, the attacker learns nearly nothing from the side channel, which leads to the same full privacy level achieved by the TDMA scheduler. The price of this added privacy is in QoS since the maximum extra queuing delay experienced by a job can be as high as $T_{acc}$.


\section{Conclusion}
\label{sec:con}

We study the information leakage through timing side channel in a job scheduler shared by a legitimate user and a malicious attacker. 
Utilizing the privacy metric defined as the equivocation of user's job arrival density, we reveal that the commonly used FCFS scheduler has a critical privacy flaw in that the attacker is able to learn exactly  user's job pattern. 
To mitigate the privacy leakage in such a scheduler, we introduce an accumulate-and-serve policy, which services jobs from the user and attacker in batches buffered during an accumulate interval. This much weakens the correlation between user's arrival process and attacker's departure process, albeit at the price of queuing delay. Our analysis indicates that full privacy can be achieved when large accumulate intervals are used.

\appendix

\subsection{Proof of Lemma~\ref{lem:stability}}
\label{app:system_dynamic}
When $\lambda+\omega<1$, the queue lengths observed at clock period boundaries $\{q(iT)\}, i=0, 1,,\cdots,$ form a positive recurrent Markov chain.
\begin{proof}
The Markovian property directly results from the FCFS policy and memoryless property of user's arrival process; 
given the queue length at time $iT$, $q(iT)$,  the future queue lengths are independent with the arrival history before $iT$. 

We show the ergodicity of this Markov chain using the linear Lyapunov function as given by 
\begin{equation}V\left (q(iT)\right)=q(iT), i=0,1,2,\cdots.\label{eq:lyapunov}\end{equation}

If $q(iT)\geq T-1$, the scheduler is guaranteed to be busy during $\left[iT, (i+1)T\right)$. Thus the queue length at time $(i+1)T$ is updated as\begin{equation}
\begin{aligned}
q\left((i+1)T\right)=q\left(iT\right)+1+a_i+x_i-T,
\label{eq:q_update}
\end{aligned}
\end{equation}
where `1' represents the {\em Type-I} attack job sent at $iT$, $a_i$ is the number of  {\em Type-II} attack jobs, and $x_i$ is the total number of user's jobs arriving during $[iT, (i+1)T)$.  $a_i$ and $x_i$ are both binomial with mean of $\omega T-1$ and $\lambda T$, respectively.
The drift of the Lyapunov function is then written by
\begin{equation}
\begin{aligned}
\mathbf{P}V(q(iT))-V(q(iT))=-(1-\omega-\lambda)T.   
\label{eq:drift_equal}
\end{aligned}
\end{equation}

Additionally, during one clock period $T$, the buffer queue length can grow at most by $T$, hence the drift is bounded by
\begin{equation}
\mathbf{P}V(q)-V(q)\leq T,  \quad  \forall q\geq 0.
\label{eq:drift_bound}
\end{equation}

Overall, combine~\eqref{eq:drift_equal} and~\eqref{eq:drift_bound}, the drift in any state satisfies
\begin{equation}
\mathbf{P}V(q)-V(q)\leq-\epsilon+TI_{\{q<T\}},
\label{eq:drift_final}
\end{equation}
where $\epsilon=(1-\omega-\lambda)T$,
and $I$ is an indicator function taking value of `1' if $q<T$. 
Following from Foster-Lyapunov stability criterion~\cite[Theorem~5]{Foster}, \eqref{eq:drift_final} implies the Markov chain $\{q(iT)\}, i=1,2,\cdots,$ is positive recurrent. 
\end{proof}

\subsection{Proof of Corollary~\ref{cor:stability}}
\label{app:system_dynamic_2}
When $\omega+\lambda<1$, the pairs 
$\left\{\tilde{A}_k, q(\tilde{A}_k); (i-1)T\leq \tilde{A}_k \leq iT\right\}, i=1,2,\cdots,$ form a positive recurrence Markov chain.  
\begin{proof}
Similar as the proof of Lemma~\ref{lem:stability} in~Appendix~\ref{app:system_dynamic}, the Markovian property directly results from the FCFS service policy. 
We only need to show the positive recurrent part. 

Notice that outgoing transitions from a state 
$\left\{\tilde{A}_k, q(\tilde{A}_k); (i-1)T\leq \tilde{A}_k \leq iT\right\}$ depend only on the last element in this state, $q(iT)$.
The transition probabilities to the next state depend on job arrival events in the next clock period $[iT, (i+1)T)$, which are homogenous among all clock periods. As a result, given the stationary distribution of $q(iT)$, the existence of which is guaranteed by  Lemma~\ref{lem:stability}, we can easily compute a 
a stationary distribution for $\left\{\tilde{A}_k, q(\tilde{A}_k); (i-1)T\leq \tilde{A}_k \leq iT\right\}$. 
The existence of a stationary distribution implies that 
 the Markov chain $\left\{\tilde{A}_k, q(\tilde{A}_k); (i-1)T\leq \tilde{A}_k \leq iT\right\}, i=1,2,\cdots,$ must be positive recurrent~\cite[Definition~3.1]{gilks95}. 
\end{proof}

\subsection{Complement of Proof of Lemma~\ref{lem:queue_length}}
\label{app:queue_length}
In this section, we analyze the stationary distribution of the queue length of the FCFS scheduler, where the attacker issues the two types of jobs as depicted in Figure~\ref{fig:attack_strategy}. 
Specifically, we study the high traffic region, where the attacker's job rate approaches its maximum, i.e.,
$\omega\to 1-\lambda$.  

\begin{lemma}
In the stationary state, queue lengths seen by {\em Type-I} attack jobs are always greater than $T-1$, i.e., 
\begin{equation}
\lim_{\omega\to 1-\lambda} Pr\left(Q_1\geq T-1\right) =1. 
\label{eq:temp5}
\end{equation}
where $Q_1$ takes the stationary distribution of states in the Markov chain  $\{q(iT)\}, i=0, 1,\cdots$. 
\label{lem:app_1}
\end{lemma}
\begin{proof}
We prove this lemma with three steps;
we first construct a `virtual' attack strategy which only issues bursty jobs on clock period boundaries. 
We next prove that the statement of the lemma holds for this virtual attack. 
Last, we show that queue length distribution in the virtual attack is dominated by our real attack defined in Figure~\ref{fig:attack_strategy}, which implies the statement in this lemma holds for the real attack.

\begin{figure}[t]
   \centering
   \includegraphics[width=0.7\columnwidth]{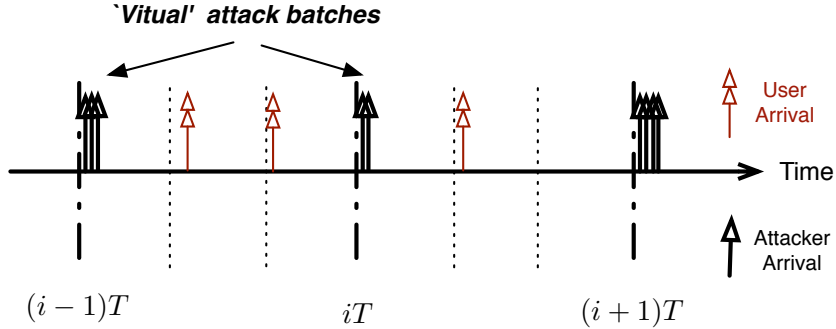} 
   \caption{
   A `virtual' attack strategy. The attacker issues jobs only on the boundaries of clock periods. At each clock tick, an amount of $1+a_i, i=0,1,\cdots$ attack jobs are sent, where $a_i\sim B\left(T-1, \frac{\omega T-1}{T-1}\right).$}
   \label{fig:virtual_attack}
\end{figure}

{\em Step 1:} Consider a virtual attack strategy that works as follows: the attacker issues a batch of bursty jobs at the beginning slot of each clock period with total number of $1+a_i, i=0,1,\cdots$, where $a_i$ is binomial $B\left(T-1, \frac{\omega T-1}{T-1}\right)$.
This attack issues the same amount of jobs in each clock period as our real attack in Figure~\ref{fig:attack_strategy}, but is not feasible in reality as the attacker cannot send more than one job in one time slot.

{\em Step 2:} Denote the queue length function under this virtual attack by $\hat{q}(\cdot)$. 
At the beginning of each clock period, the queue  length updates as
\begin{equation}
\hat{q}((i+1)T)=\left ( \hat{q}(iT)+1+ a_i +X_i-T \right )_+, \quad i=0,1,\cdots.
\label{eq:virtual_queue}
\end{equation}
Using the same Lyapunov function as we define in \eqref{eq:lyapunov} of the proof of Lemma~\ref{lem:stability}, it is not hard show that  $\{\hat{q}(iT)\}, i=0, 1,\cdots$ form a positive recurrent Markov chain when $\lambda+\omega<1$. 
Define $\hat{Q}_1$ as a random variable with the stationary distribution of this chain, we now prove
\begin{equation}
\lim_{\omega\to 1-\lambda} Pr\left(\hat{Q}_1\geq T-1\right) =1
\label{eq:temp4}
\end{equation}
using the $z$-transform of sequence $\{\hat{q}(iT)\}$,  derived from~\eqref{eq:virtual_queue} as given by
\begin{equation}
\begin{aligned}
&\hat{\mathcal{Q}}(z)=\frac{\sum_{k=0}^{T-2}\sum_{r=0}^{T-2-k}\sum_{o=0}^{T-2-k-r}p_k u_r v_o(z^{T-1}-z^{k+r+o})}{z^{T-1}-\mathcal{A}(z)\mathcal{X}(z)}
\label{eq:z-trans}
\end{aligned}
\end{equation}
where  $p_k=Pr\left(\hat{Q}_1=k\right)$, $u_r=Pr\left(a_i=r\right)$, and $v_o=Pr\left(X_i=o\right)$. Moreover,
$\mathcal{A}(z)$ and $\mathcal{X}(z)$ are the $z$-transforms of sequence $\{a_i\}$ and $\{X_i\}$, and \begin{equation}
\mathcal{A}(z)=\left(1- \frac{\omega T-1}{T-1}+ \frac{\omega T-1}{T-1}z\right)^{T-1}\label{eq:z_a}
\end{equation} and \begin{equation} \mathcal{X}(z)=(1-\lambda+\lambda z)^{T}.\label{eq:z_x}\end{equation}

Subsituting~\eqref{eq:z_a},~\eqref{eq:z_x} into~\eqref{eq:z-trans} and taking $z=1$ on both sides, we get
\begin{equation}
\begin{aligned}
\sum_{k=0}^{T-2} p_k \cdot \left(  \sum_{r=0}^{T-2-k}\sum_{o=0}^{T-2-k-r} u_rv_o(T-1-(k+r+o)) \right)
=T(1-\omega-\lambda).
\label{eq:temp1}
\end{aligned}
\end{equation}
Dropping the terms with $r>0$ or $o>0$ on the left hand side of the equality, we further get
\begin{equation}
\begin{aligned}
 u_0v_0(T-1-k) \sum_{k=0}^{T-2} p_k \leq T(1-\omega-\lambda). 
\label{eq:temp2}
\end{aligned}
\end{equation}
Plugging in the values of $u_0$ and $v_0$ in~\eqref{eq:temp2}, 
\begin{equation}
 \sum_{k=0}^{T-2}p_k \leq  \frac{(T-1)^{T-1}(1-\omega-\lambda)}{T^{T-2}(1-\omega)^{T-1}(1-\lambda)^{T}} 
\label{eq:temp3}
\end{equation}
Taking the limit $\omega\to 1-\lambda$, we get
\begin{equation}\lim_{\omega\to 1-\lambda}\sum_{k=0}^{T-2} p_k \leq \lim_{\omega\to 1-\lambda}  \frac{(T-1)^{T-1}(1-\omega-\lambda)}{T^{T-2}(1-\omega)^{T-1}(1-\lambda)^{T}} = 0.\end{equation}
This completes the proof of~\eqref{eq:temp4}.

{\em Step 3:} 
We next extend~\eqref{eq:temp4} to the case of our real attack, based on the fact that  the queuing process in the real attack dominates the queuing process in the virtual attack (See Lemma~\ref{lem:dominance} for the proof). 
Define $Q_1$ as a random variable taking the stationary distribution of states $\{q(iT)\}, i=0,1,\cdots$. Lemma~\ref{lem:dominance} tells us 
\begin{equation}
\lim_{\omega\to 1-\lambda} Pr\left(Q_1\geq T-1\right) \geq \lim_{\omega\to1-\lambda} Pr(\hat{Q}_1\geq T-1).
\label{eq:temp10}
\end{equation}
Plug~\eqref{eq:temp10} into~\eqref{eq:temp4}, \eqref{eq:temp5} is proved. 
\end{proof}

\begin{lemma}
The stationary distribution of the Markov chain $\{\hat{q}(iT)\}, i=0,1,\cdots$ in the virtual 
attack is dominated by the stationary distribution of the  Markov chain $\{q(iT)\}, i=0,1,\cdots$ in the real attack; i.e.,
\begin{equation}
Pr\left(Q_1\geq q\right)\geq Pr\left(\hat{Q}_1\geq q\right), \quad  \forall q\geq 0,
\label{eq:dominance}
\end{equation}
where $Q_1$ and $\hat{Q}_1$ are random variables taking the stationary distributions of $\{q(iT)\}, i=0,1,\cdots$ and $\{\hat{q}(iT)\}, i=0,1,\cdots$, respectively. 
\label{lem:dominance}
\end{lemma}

\begin{proof}
Recall in the real attack strategy, the queue length seen by each attacker's job updates as
\begin{equation}
q(\tilde{A}_{j+1})=\left(q(\tilde{A}_{j})+1+\hat{X}_{j}-(\tilde{A}_{j+1}-\tilde{t}_{j})\right)_{+}, j=1,2,\cdots.
\label{eq:temp15}
\end{equation}
where $\hat{X}_k$ is the number of user's jobs arriving between $\tilde{A}_j$ and $\tilde{A}_{j+1}$. 
Consider~\eqref{eq:temp15} for $j$ taking values from $k$ to $r-1$, and sum up all the resulting equations, we derive the inequality that
\begin{equation}
q(\tilde{A}_{r})\geq \left(q(\tilde{A}_{k})+r-k+1+\sum_{j=k}^{r-1}\hat{X}_{j}-(\tilde{A}_{r}-\tilde{A}_{k})\right)_{+}.
\label{eq:temp6}
\end{equation}

Now make $k=\inf\{j:iT\leq \tilde{A}_j\leq (i+1)T\}$ and $r=\sup\{j:iT\leq \tilde{A}_j\leq (i+1)T\}$, i.e., indices of the attacker's jobs sent at time $iT$ and $(i+1)T$, we get
\begin{equation}
q((i+1)T) \geq  \left( q(iT) +1+ a_i +X_i -T \right)_+ 
\label{eq:temp7}
\end{equation}
where  $X_i=\underset{j:iT\leq \tilde{A}_j< (i+1)T}{\sum} \hat{X}_j$ is the number of user's jobs arriving in the $i^{th}$ clock period, and 
$a_i=|\{j:iT< \tilde{A}_j<(i+1)T\}|$ is the number of {\em Type-II} jobs sent by the attacker in the $i^{th}$ clock period.

Compare~\eqref{eq:temp7} with the queue length update equation for the virtual attack in~\eqref{eq:virtual_queue}, we can show by induction that  ${q}(iT)\geq \hat{q}(iT)$ for $i=1,2,\cdots$, assuming $q(0)=\hat{q}(0)=0$, which implies~\eqref{eq:dominance}.
\end{proof}

\subsection{Continuation of proof of Theorem~\ref{theo:acc_bound}}
\label{app:complement}
\begin{equation}
\mathcal{P}^{T}\geq \left(1-\frac{T}{T_{acc}}+\frac{T}{T_{L}}\right)H(X) - \frac{TH\left(\sum_{i=1}^{\left\lfloor\frac{T_{acc}}{T}\right\rfloor}X_i\right)}{T_{acc}}
\label{eq:acc_bound_0}
\end{equation}
where $T_{L}=lcm\left(T,T_{acc}\right)$, where $X, X_1,\cdots,X_{\left\lfloor\frac{T_{acc}}{T}\right\rfloor}$ are are i.i.d. binomial $B(T,\lambda).$

We give the proof of~\eqref{eq:acc_bound_0} when $T$ does not divide $T_{acc}$. 

\begin{proof}
From Lemma~\ref{lem:acc_chain}, the privacy of the accumulate-and-serve scheduler is lower-bounded by
\begin{equation}
\mathcal{P}^{T}\geq \lim_{n\to \infty} \frac{H\left(\mathbf{X}^{\frac{nT_L}{T}}|\mathbf{Z}^{\frac{nT_{L}}{T_{acc}}}\right)}{\frac{nT_{L}}{T}}
\label{eq:acc_bound_3}
\end{equation}
where $T_{L}=lcm(T_{acc},T)$.

\begin{figure*}[t]
   \centering
   \includegraphics[width=\textwidth]{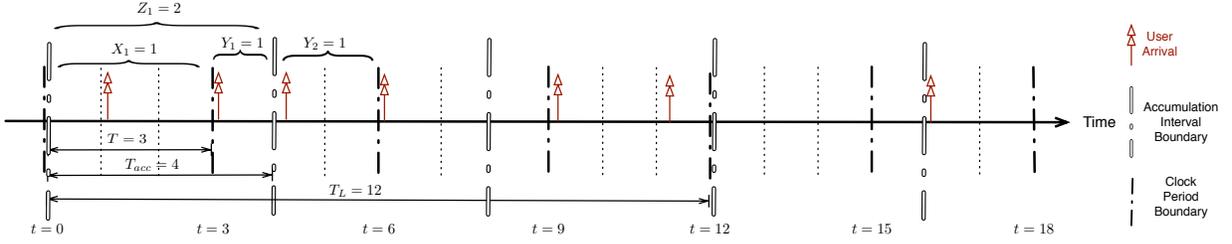} 
   \caption{Illustration of the accumulate-and-serve scheduler when $T=3$ and $T_{acc}=4$. The boundaries of the clock periods and accumulate intervals overlap every $12$ time slots. Some job pattern is split by the accumulate interval boundaries into two pieces, e.g., $X_2=Y_1+Y_2.$}
   \label{fig:acc_proof_1}
\end{figure*}

Notice that the clock period boundaries and accumulate interval boundaries  overlap every $T_L$ time slots; i.e, 
\begin{equation}
\sum_{j=\frac{iT_{L}}{T}+1}^{\frac{(i+1)T_{L}}{T}}X_j = \sum_{j=\frac{iT_{L}}{T_{acc}}+1}^{\frac{(i+1)T_{L}}{T_{acc}}}Z_j, \quad i=0,1,\cdots.
\end{equation}
Thus, \eqref{eq:acc_bound_3} can be rewritten as 
\begin{equation}
\begin{aligned}
\mathcal{P}^{T}\geq \lim_{n\to \infty} \frac{\sum_{i=0}^{n-1}H\left(\mathbf{X}_{\frac{iT_{L}}{T}+1}^{\frac{(i+1)T_{L}}{T}}\Big|\mathbf{Z}_{\frac{iT_{L}}{T_{acc}}+1}^{\frac{(i+1)T_{L}}{T_{acc}}}\right)}{\frac{nT_{L}}{T}}\overset{(a)}{=}\frac{H\left(\mathbf{X}^{\frac{T_{L}}{T}}\big| \mathbf{Z}^{\frac{T_{L}}{T_{acc}}}\right)}{\frac{T_L}{T}}
\end{aligned}
\label{eq:acc_bound_4}
\end{equation}
where $(a)$ follows from the fact that both sequence $\mathbf{X}$ and $\mathbf{Z}$ are i.i.d. binomial random variables. 

Among the first $\frac{T_L}{T}$ clock periods, $\frac{T_{L}}{T_{acc}}-1$ clock periods lie across two accumulate intervals.
For example, in Figure~\ref{fig:acc_proof_1}, where $T=3$ and $T_{acc}=4$, the 
$2^{nd}$ and $3^{rd}$ $T$-clock period cross two accumulate intervals. 
Denote $\mathbf{Y}^{2\left(\frac{T_{L}}{T_{acc}}-1\right)}$ to be the number of users's jobs in split clock periods, 
\begin{equation}
Y_{2(j-1)+1}+Y_{2j}=X_{\left \lceil \frac{jT_{acc}}{T} \right\rceil}, \quad j=1,2\cdots,\frac{T_{L}}{T_{acc}}-1,
\end{equation}
and assign $\mathbf{Y}^{2\left(\frac{T_{L}}{T_{acc}}-1\right)}$ to the  attacker as extra information,  we get
\begin{equation}
\begin{aligned}
H\left(\mathbf{X}^{\frac{T_{L}}{T}}\big| \mathbf{Z}^{\frac{T_{L}}{T_{acc}}}\right)& \geq H\left(\mathbf{X}^{\frac{T_{L}}{T}}\Big| \mathbf{Z}^{\frac{T_{L}}{T_{acc}}}, \mathbf{Y}^{2\left(\frac{T_{L}}{T_{acc}}-1\right)}\right)  \\
&\overset{(b)}{=} \sum_{j=1}^{\frac{T_L}{T_{acc}}} H\left( \mathbf{X}^{\left \lceil \frac{jT_{acc}}{T} \right\rceil-1}_{\left \lceil \frac{(j-1)T_{acc}}{T} \right\rceil+1} \Big| Z_j, Y_{2j}, Y_{2j+1} \right)\\
&\overset{(c)}= \sum_{j=1}^{\frac{T_L}{T_{acc}}} H\left( \mathbf{X}^{\left \lceil \frac{jT_{acc}}{T} \right\rceil-1}_{\left \lceil \frac{(j-1)T_{acc}}{T} \right\rceil+1} \Big| \sum^{\left \lceil \frac{jT_{acc}}{T} \right\rceil-1}_{k=\left \lceil \frac{(j-1)T_{acc}}{T} \right\rceil+1} X_k\right)\\
&\overset{(d)}{=} \sum_{j=1}^{\frac{T_L}{T_{acc}}}\left( H\left(\mathbf{X}^{\left \lceil \frac{jT_{acc}}{T} \right\rceil-1}_{\left \lceil \frac{(j-1)T_{acc}}{T} \right\rceil+1}\right)- H\left(\sum^{\left \lceil \frac{jT_{acc}}{T} \right\rceil-1}_{k=\left \lceil \frac{(j-1)T_{acc}}{T} \right\rceil+1} X_k\right)\right)\\
&\overset{(e)}{=} \left(\frac{T_L}{T}-\frac{T_L}{T_{acc}}+1 \right) H(X)- \sum_{j=1}^{\frac{T_L}{T_{acc}}}H\left(\sum^{\left \lceil \frac{jT_{acc}}{T} \right\rceil-1}_{k=\left \lceil \frac{(j-1)T_{acc}}{T} \right\rceil+1} X_k\right)\\
&\overset{(f)}{\geq} \left(\frac{T_L}{T}-\frac{T_L}{T_{acc}}+1 \right) H(X)- \frac{T_L}{T_{acc}}H\left(\sum_{k=1}^{\left
\lfloor\frac{T_{acc}}{T}\right\rfloor}X_k\right)\\
\end{aligned}
\label{eq:acc_bound_5}
\end{equation}
where $X, X_1,\cdots, X_{\left
\lfloor\frac{T_{acc}}{T}\right\rfloor}$ are i.i.d. $B(T,\lambda)$, $(b)$ applies the chain rule and dependencies between $\mathbf{X}$ and $\mathbf{Z},\mathbf{Y}$,
$(c)$ follows from
\begin{equation}
Z_j= Y_{2j}+Y_{2j+1}+\sum^{\left \lceil \frac{jT_{acc}}{T} \right\rceil-1}_{k=\left \lceil \frac{(j-1)T_{acc}}{T} \right\rceil+1} X_k, \quad j=1,2,\cdots, \frac{T_L}{T_{acc}},
\end{equation}
$(d)$  results from the fact that variables in sequence $\mathbf{X}$ are i.i.d., $(e)$ makes use of 
\begin{equation}
 \sum_{j=1}^{\frac{T_L}{T_{acc}}} \left({\left \lceil \frac{jT_{acc}}{T} \right\rceil} - \left \lceil \frac{(j-1)T_{acc}}{T} \right\rceil -1\right) =\left(\frac{T_L}{T}-\frac{T_L}{T_{acc}}+1 \right),
\end{equation} 
and $(f)$  follows from
\begin{equation}
{\left \lceil \frac{jT_{acc}}{T} \right\rceil} - \left \lceil \frac{(j-1)T_{acc}}{T} \right\rceil -1 \leq \left \lfloor \frac{T_{acc}}{T} \right\rfloor.
\end{equation}

Substituting~\eqref{eq:acc_bound_5} back in~\eqref{eq:acc_bound_4}, \eqref{eq:acc_bound_0} is proved.
\end{proof}

\bibliographystyle{IEEETran}
\bibliography{xun}

\end{document}